\documentclass[11pt]{amsart}
\usepackage{mathrsfs}
\usepackage{amssymb} %Extra symbols 
\usepackage{url}
\usepackage{hyperref}
\usepackage{graphicx} %Include postscript graphics

\theoremstyle{definition}
\newtheorem{theorem}{Theorem}[section]

\newtheorem{corollary}[theorem]{Corollary}
\newtheorem{remark}[theorem]{Remark}

\newtheorem{openproblem}[theorem]{Open Problem}

\newcommand{\sref}[1]{\S\ref{#1}}

\newcommand{\R}{\mathbb{R}}

\newcommand{\C}{\mathbb{C}}

\newcommand{\de}{\textnormal{d}}

\newcommand{\tn}{\textnormal}
\newcommand{\ds}{\displaystyle}

\newcommand{\ie}{\textit{i.e.} }
\newcommand{\cf}{\textit{cf.} }
\newcommand{\eg}{\textit{e.g.} }
\newcommand{\citep}[2]{\cite{#1}, p. #2}

\newcommand{\mc}[1]{\mathcal{#1}}

\newcommand{\image}[3]{\begin{figure*}[ht]
\includegraphics[width=#2\textwidth]{#1}
\caption{\small{\label{#1}#3}}\end{figure*}}

\newcommand{\dsfrac}[2]{\ds{\frac{#1}{#2}}}

\newcommand{\schw}{Schwarzschild}
\newcommand{\rn}{Reissner-Nordstr\"om}

\def\hyph{-\penalty0\hskip0pt\relax}
\newcommand{\semiriem}{semi{\hyph}Riemannian}

\newcommand{\semireg}{semi{\hyph}regular}

\begin{document}

\title{Analytic Reissner-Nordstr\"om Singularity}

\author{Ovidiu-Cristinel \ Stoica}
\thanks{\footnotesize{E\,m\,a\,i\,l\,:\, h\,o\,l\,o\,t\,r\,o\,n\,i\,x\,@\,g\,m\,a\,i\,l\,.\,c\,o\,m}}
\begin{abstract}
An analytic extension of the Reissner-Nordstr\"om solution at and beyond the singularity is presented. The extension is obtained by using new coordinates in which the metric becomes degenerate at $r=0$. The metric is still singular in the new coordinates, but its components become finite and smooth. Using this extension it is shown that the charged and non-rotating black hole singularities are compatible with the global hyperbolicity and with the conservation of the initial value data. Geometric models for  electrically charged particles are obtained.
\end{abstract}

%\pacs{04.70.Bw, 04.20.-q, 98.80.-k, 02.40.Ky}
%\submitto{Physica Scripta}
\maketitle

\setcounter{tocdepth}{1}
\tableofcontents

%~~~~~~~~~~~~~~~~~~~~~~~~~~~~~~~~~~~~~~~~~~~~~~~~~~~~~~~~~~~~~~~~~~~~~~~%
\section*{Introduction}

%~~~~~~~~~~~~~~~~~~~~~~~~~~~~~~~~~~~~~~~~~~~~~~~~~~~~~~~~~~~~~~~~~~~~~~~%
\subsection{The {\rn} solution}
\label{s_rn_solution}

The {\rn} metric describes a static, spherically symmetric, electrically charged, non-rota\-ting black hole \cite{reiss16,nord18}. It is a solution to the Einstein-Maxwell equations. It has the following form:
\begin{equation}
\label{eq_rn_metric}
\de s^2 = -\left(1-\dsfrac{2m}{r} + \dsfrac{q^2}{r^2}\right)\de t^2 + \left(1-\dsfrac{2m}{r} + \dsfrac{q^2}{r^2}\right)^{-1}\de r^2 + r^2\de\sigma^2,
\end{equation}
where $q$ is the electric charge of the body and, as in the case of the {\schw} solution,
\begin{equation}
\label{eq_sphere}
\de\sigma^2 = \de\theta^2 + \sin^2\theta \de \phi^2
\end{equation}
is the metric of the unit sphere $S^2$, $m$ the mass of the body, and the units were chosen so that $c=1$ and $G=1$ (see, \eg, \citep{HE95}{156}).

The first two terms in the right hand side of equation \eqref{eq_rn_metric} are independent on the coordinates $\theta$ and $\phi$, and conversely, $\de\sigma^2$ is independent on the coordinates $r$ and $t$. This solution is a warped product between a two-dimensional ($2D$) {\semiriem} space and the sphere $S^2$ with the canonical metric \eqref{eq_sphere}. Consequently, in coordinate transformations which affect only the coordinates $r$ and $t$ we can ignore the term $r^2\de\sigma^2$ in calculations. We can, finally, reintroduce it, taking again the warped product.

Solving the equation expressing the cancellation of $r^2 - 2mr + q^2$ for $r$ we obtain:
\begin{enumerate}
	\item no solution, for $q^2 > m^2$ (naked singularity);
	\item double solution $r_\pm = m$, for $q^2 = m^2$ (the \textit{extremal} case);
	\item two solutions $r_\pm = m \pm \sqrt{m^2 - 2^2}$, for $q^2<m^2$.
\end{enumerate}
If $q^2 \leq m^2$, there are two singular horizons at $r=r_\pm$, which coincide for $q^2 = m^2$. These apparent singularities can be removed by a special coordinate transformation, such as that of Eddington-Finkelstein. All three cases have an irremovable singularity at $r=0$.

%~~~~~~~~~~~~~~~~~~~~~~~~~~~~~~~~~~~~~~~~~~~~~~~~~~~~~~~~~~~~~~~~~~~~~~~%
\subsection{Two kinds of metric singularity}
\label{s_two_kinds_singularity}

There are two main kinds of metric singularities that are relevant for our approach. In the first kind, some of the components of the metric diverge as approaching the singularity, where they become infinite. In the second kind, more benign, the metric remains always smooth, but becomes \textit{degenerate} at the singularity -- that is, its determinant becomes $0$. In general, this means that the metric is not invertible, \ie $g^{ab}$ tends to infinity and cannot be defined at the singularity. But $g_{ab}$ is smooth, and in some cases, despite the fact that $g^{ab}$ is singular, we can define the contraction between covariant indices, and construct covariant derivatives and the Riemann curvature, and even write an equivalent of the Einstein equation \cite{Sto11a,Sto11b,Sto11d}.

Some singularities of the first kind, having some components of the metric divergent, can be viewed as singularities of the second kind, expressed in singular coordinate systems. This means that it is possible for some singularities of the first kind to be transformed into singularities of the second kind, by an appropriate choice of the coordinate system. We did this for the {\schw} solution in \cite{Sto11e}.

In this paper, we will find for the {\rn} solution a new coordinate system, in which the singularity at $r=0$ becomes degenerate and analytic. The metric will become degenerate, but all its coefficients will be finite and smooth. The new form of the metric admits an analytic continuation beyond the singularity.

%~~~~~~~~~~~~~~~~~~~~~~~~~~~~~~~~~~~~~~~~~~~~~~~~~~~~~~~~~~~~~~~~~~~~~~~%
\section{Extending the {\rn} spacetime at the singularity}
\label{s_rn_ext_ext}

%~~~~~~~~~~~~~~~~~~~~~~~~~~~~~~~~~~~~~~~~~~~~~~~~~~~~~~~~~~~~~~~~~~~~~~~%
\subsection{The main result}
\label{s_rn_ext_ext_central}

The main result of this paper is contained in the following theorem.

\begin{theorem}
\label{thm_rn_ext_ext}
The {\rn} metric admits an analytic extension at $r=0$.
\end{theorem}
\begin{proof}
We work initially in two dimensions $(t,r)$. 
It is enough to make the coordinate transformation in a neighborhood of the singularity -- in the region $r\in[0,M)$, where $M=r_-$ if $q^2 \leq m^2$, and $M=\infty$ otherwise. We choose the coordinates $\rho$ and $\tau$, so that
\begin{equation}
\label{eq_coordinate_ext_ext}
\begin{array}{l}
\Bigg\{
\begin{array}{ll}
t &= \tau\rho^T \\
r &= \rho^S \\
\end{array}
\\
\end{array}
\end{equation}
where $S,T$ have to be determined in order to make the metric analytic (figure \ref{coordinates}). This choice is motivated by the need to stretch the spacetime while approaching the singularity $r=0$, so that the divergent components of the metric are smoothened.

\image{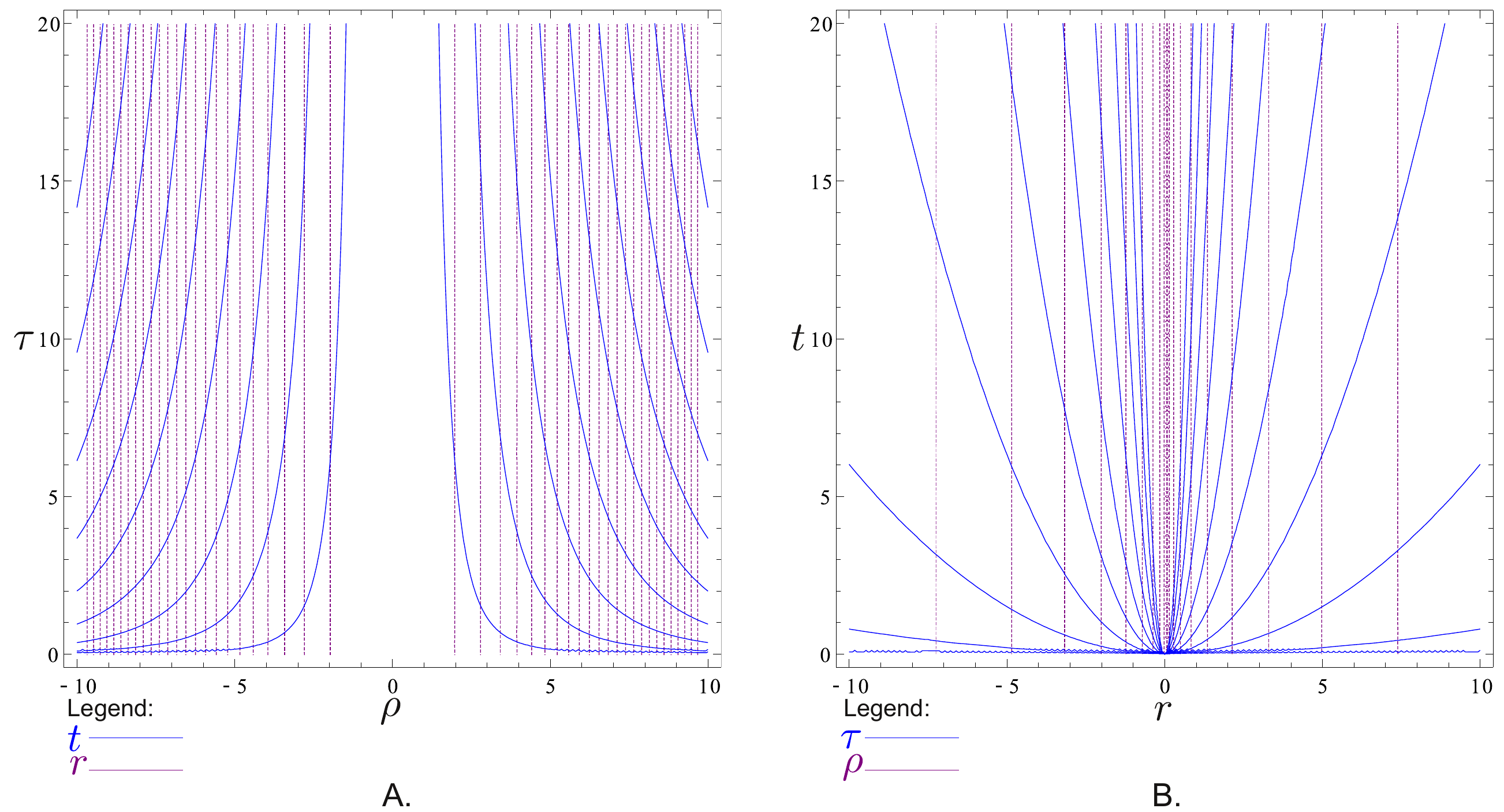}{0.85}{The coordinate transformations \eqref{eq_coordinate_ext_ext}, represented for $S=2$ and $T=4$. A. The coordinates $(t,r)$ expressed in coordinates $(\tau,\rho)$. B. The coordinates $(\tau,\rho)$ expressed in coordinates $(t,r)$.}

Then, we have
\begin{equation}
\label{eq_coordinate_jacobian}
\dsfrac{\partial t}{\partial \tau} = \rho^T,\,
\dsfrac{\partial t}{\partial \rho} = T\tau\rho^{T-1},\,
\dsfrac{\partial r}{\partial \tau} = 0,\,
\dsfrac{\partial r}{\partial \rho} = S\rho^{S-1}.
\end{equation}
Let us introduce the standard notation
\begin{equation}
	\Delta := r^2 - 2m r + q^2 \tn{ (hence $\Delta = \rho^{2S} - 2m \rho^{S} + q^2$)}.
\end{equation}
We note that $\Delta> 0$ for $\rho\in[0,M)$.

The metric components in \eqref{eq_rn_metric} become now
\begin{equation}
\label{eq_metric_coeff_schw}
g_{tt} = - \dsfrac{\Delta}{\rho^{2S}},\,
g_{rr} = \dsfrac{\rho^{2S}}{\Delta},\,
g_{tr} = g_{rt} = 0.
\end{equation}
Let us calculate the metric components in the new coordinates. 
\begin{equation*}
\begin{array}{lll}
g_{\tau\tau} &=& \left(\dsfrac{\partial r}{\partial \tau}\right)^2\dsfrac{\rho^{2S}}{\Delta} - \left(\dsfrac{\partial t}{\partial \tau}\right)^2\dsfrac{\Delta}{\rho^{2S}} \\
&=& 0 - \rho^{2T}\dsfrac{\Delta}{\rho^{2S}} \\
\end{array}
\end{equation*}
Therefore
\begin{equation}
\label{eq_metric_coeff_ext_ext_tau_tau}
g_{\tau\tau} = -\Delta\rho^{2T-2S}.
\end{equation}
\begin{equation*}
\begin{array}{lll}
g_{\rho\tau} &=& \dsfrac{\partial r}{\partial \rho}\dsfrac{\partial r}{\partial \tau}\dsfrac{\rho^{2S}}{\Delta} - \dsfrac{\partial t}{\partial \rho}\dsfrac{\partial t}{\partial \tau}\dsfrac{\Delta}{\rho^{2S}} \\
&=& 0 - T\tau\rho^{2T-1}\dsfrac{\Delta}{\rho^{2S}}
\end{array}
\end{equation*}
Then
\begin{equation}
\label{eq_metric_coeff_ext_ext_rho_tau}
g_{\rho\tau} = - T\Delta\tau\rho^{2T-2S-1}.
\end{equation}
\begin{equation*}
\begin{array}{lll}
g_{\rho\rho} &=& S^2\rho^{2S-2}\dsfrac{\rho^{2S}}{\Delta} - T^2\tau^2\rho^{2T-2}\dsfrac{\Delta}{\rho^{2S}} \\
\end{array}
\end{equation*}
Hence
\begin{equation}
\label{eq_metric_coeff_ext_ext_rho_rho}
g_{\rho\rho} = S^2\dsfrac{\rho^{4S-2}}{\Delta} - T^2\Delta\tau^2\rho^{2T-2S-2}.
\end{equation}

We can see from the term $S^2\dsfrac{\rho^{4S-2}}{\Delta}$ of $g_{\rho\rho}$ from equation \eqref{eq_metric_coeff_ext_ext_rho_rho} that, to ensure that the singularity at $\rho=0$ is only of degenerate type and the metric is continuous there, $S$ has to be an integer so that $S\geq 1$. Moreover, the condition $S\geq 1$ makes this term analytic at $\rho=0$, because the denominator does not cancel there and is analytic, and the numerator is analytic.

The other terms of the equation \eqref{eq_metric_coeff_ext_ext_rho_rho}, and the other equations \eqref{eq_metric_coeff_ext_ext_tau_tau} and \eqref{eq_metric_coeff_ext_ext_rho_tau}, contain as factor $\Delta$, in which the minimum power to which $\rho$ appears is $0$. Hence, in order to avoid negative powers of $\rho$, these terms require that $2T - 2S - 2 \geq 0$. Therefore, the conditions for removing the infinity of the metric at $r=0$ by a coordinate transformation are that $S$ and $T$ be integers so that:
\begin{equation}
\label{eq_metric_smooth_cond}
\begin{array}{l}
\Bigg\{
\begin{array}{ll}
S \geq 1 \\
T \geq S + 1	
\end{array}
\\
\end{array}
\end{equation}
and they also ensure that the metric is analytic at $r=0$. None of the metric's components become infinite at the singularity.

To go back to four dimensions, we have to take the warped product between the $2D$ space with the metric we obtained, and the sphere $S^2$, with warping function $\rho^{S}$. This is a degenerate warped product, as was studied in \cite{Sto11b}, and its result is a $4D$ manifold whose metric is analytic and degenerate at $\rho=0$. 
Hence, this extension of the {\rn} solution is analytic at $\rho=0$.
\end{proof}

Let us extract from the proof the expression of the metric:
\begin{corollary}
The {\rn} metric, expressed in the coordinates from theorem \ref{thm_rn_ext_ext}, has the following form
\begin{equation}
\label{eq_rn_ext_ext}
\de s^2 = - \Delta\rho^{2T-2S-2}\left(\rho\de\tau + T\tau\de\rho\right)^2 + \dsfrac{S^2}{\Delta}\rho^{4S-2}\de\rho^2 + \rho^{2S}\de\sigma^2.
\end{equation}
\end{corollary}
\begin{proof}
From \eqref{eq_coordinate_jacobian} we find
\begin{equation}
\label{eq_de_t}
	\de t = \dsfrac{\partial t}{\partial \tau}\de\tau + \dsfrac{\partial t}{\partial \rho}\de\rho = \rho^T\de\tau + T\tau\rho^{T-1}\de\rho = \rho^{T-1}(\rho\de\tau + T\tau\de\rho)
\end{equation}
and
\begin{equation}
\label{eq_de_r}
	\de r = \dsfrac{\partial r}{\partial \tau}\de\tau + \dsfrac{\partial r}{\partial \rho}\de\rho = S\rho^{S-1}\de\rho
\end{equation}
which when plugged in the {\rn} equation \eqref{eq_rn_metric} give
\begin{equation*}
\begin{array}{lll}
	\de s^2 &=& - \dsfrac{\Delta}{\rho^{2S}}\de t^2 + \dsfrac{\rho^{2S}}{\Delta}\de r^2 + r^2\de\sigma^2 \\
	&=& -\Delta\rho^{2T-2S-2}(\rho\de\tau + T\tau\de\rho)^2 + \dsfrac{S^2}{\Delta}\rho^{4S-2}\de\rho^2 + \rho^{2S}\de\sigma^2.
\end{array}
\end{equation*}
\end{proof}

%~~~~~~~~~~~~~~~~~~~~~~~~~~~~~~~~~~~~~~~~~~~~~~~~~~~~~~~~~~~~~~~~~~~~~~~%
\subsection{The electromagnetic field}
\label{s_rn_ext_ext_electromagnetic}

The potential of the electromagnetic field in the {\rn} solution is
\begin{equation}
A = -\dsfrac q r \de t,
\end{equation}
and is singular at $r=0$ in the standard coordinates $(t,r,\phi,\theta)$. On the other hand, in the new coordinates it is smooth.

\begin{corollary}
\label{rem_rn_electromagnetic_potential}
In the new coordinates $(\tau,\rho,\phi,\theta)$, the electromagnetic potential is
\begin{equation}
A = -q\rho^{T-S-1}\left(\rho\de\tau + T\tau\de\rho\right),
\end{equation}
the electromagnetic field is
\begin{equation}
F = q(2T-S)\rho^{T-S-1}\de\tau \wedge\de\rho,
\end{equation}
and they are analytic everywhere, including at the singularity $\rho=0$.
\end{corollary}
\begin{proof}
The equation of the electromagnetic potential follows directly from the proof of theorem \ref{thm_rn_ext_ext} and from equation \eqref{eq_de_t}. The equation of the electromagnetic field is obtained by applying the exterior derivative:
\begin{equation*}
\begin{array}{lll}
F &=& \de A = -q\de\left(\rho^{T-S}\de\tau + T\tau\rho^{T-S-1}\de\rho\right) \\
&=& -q \left(\dsfrac{\partial\rho^{T-S}}{\partial \rho} \de\rho\wedge\de\tau + T\dsfrac{\partial\tau\rho^{T-S-1}}{\partial\tau}\de\tau\wedge\de\rho\right) \\
&=& -q \left((T-S)\rho^{T-S-1} \de\rho\wedge\de\tau + T \rho^{T-S-1}\de\tau\wedge\de\rho\right) \\
&=& q(2T-S)\rho^{T-S-1}\de\tau\wedge\de\rho
\end{array}
\end{equation*}
\end{proof}

%~~~~~~~~~~~~~~~~~~~~~~~~~~~~~~~~~~~~~~~~~~~~~~~~~~~~~~~~~~~~~~~~~~~~~~~%
\subsection{General remarks concerning the proposed degenerate extension}
\label{s_rn_ext_remarks}

\begin{remark}
\label{rem_rn_wormhole}
The analytic {\rn} solution we found extends through $\rho=0$ to negative values of $\rho$.
If $S$ is even, $\rho$ and $-\rho$ give the same metric. For even values of $T$, also the electromagnetic potential is invariant at the space inversion $\rho\mapsto-\rho$. After taking the warped product we can identify the points $(\tau,\rho,\phi,\theta)$ and $(\tau,-\rho,(\phi+\pi)\mod 2\pi,\theta)$, and also all the points from the warped product which have $\rho=0$ and constant $\tau$. This identification gives a smooth metric, because of the symmetry with respect to the axis $\rho=0$, and because the warping function is $\rho^{S}$, with $S\geq 1$. We obtain by this a spherically symmetric solution having the topology of $\R^4$.

If we choose not to make this identification, the extension through $\rho=0$ looks like the Einstein-Rosen model of charged particles \cite{ER35}, or like Misner and Wheeler's ``charge without charge'' \cite{MW57}. As is known from the ``charge without charge'' program, special topology (\ie ``wormholes'') allows the existence of source-free electromagnetic fields which look as being associated to charges, without actually having sources. The proposed degenerate extension of the {\rn} spacetime seems to support these proposals, but by making the above-mentioned identification, it also allows charge models with the standard $\R^4$ topology.

If $S$ is odd, the extension to $\rho<0$ is very similar to the extension from the Kerr and Kerr-Newmann solutions through the interior of the ring singularity to the region $r<0$.
\end{remark}

\begin{remark}
As in the case of the analytic extension of the {\schw} solution \cite{Sto11e}, there is no unique way to extend the {\rn} metric so that it is smooth at the singularity. The explanation is due to the fact that a degenerate metric can remain smooth and even analytic at certain singular coordinate transformations.
\end{remark}

A \textit{\semireg} metric has smooth Riemann curvature $R_{abcd}$, and allows the construction of more useful operations which are normally prohibited by the fact that the metric is degenerate. In the case of the {\schw} black hole we could find a solution which is {\semireg} \cite{Sto11e}. In the case of the {\rn} black hole, we can't find numbers $S$ and $T$ for the equation \eqref{eq_coordinate_ext_ext}, which would make the metric \semireg. However, this does not exclude other changes of the coordinates, and we propose the following open problem:
\begin{openproblem}
Is it possible to find coordinates which allow the {\rn} metric to be {\semireg} at $\rho=0$?
\end{openproblem}
Also, it may be interesting the following:
\begin{openproblem}
Can we find natural conditions ensuring the uniqueness of the analytic extensions of the {\schw} and {\rn} solutions at the singularity $\rho=0$? Under what conditions does a singular coordinate transformation of an analytic extension lead to another extension which is physically indistinguishable?
\end{openproblem}

%~~~~~~~~~~~~~~~~~~~~~~~~~~~~~~~~~~~~~~~~~~~~~~~~~~~~~~~~~~~~~~~~~~~~~~~%
\section{Null geodesics in the proposed solution}
\label{s_rn_ext_null_geodesics}

In this section, we will discuss the geometric meaning of the extension proposed in this paper, mainly from the viewpoint of the lightcones and the null geodesics. In the coordinates $(\tau,\rho)$, the metric is analytic near the singularity $\rho=0$ and has the form
\begin{equation}
\label{eq_rn_metric_tau_rho_matrix}
g = -\Delta\rho^{2T-2S-2}\left(
\begin{array}{ll}
    \rho^2 & T\tau\rho \\
    T\tau\rho & T^2\tau^2 - \dsfrac{S^2}{\Delta^2}\rho^{6S-2T} \\
\end{array}
\right)
\end{equation}

Let us find the null directions, defined at each point $(\tau,\rho)$ by the tangent vectors $u\neq 0$ so that $g(u, u)=0$. Since any nonzero multiple of $u$ is also a solution, we will consider $u=(\sin\alpha,\cos\alpha)$, and try to find $\alpha$. We obtain the equation
\begin{equation}
	\rho^2\sin^2\alpha + 2T\tau\rho\sin\alpha\cos\alpha + \left(T^2\tau^2 - \dsfrac{S^2}{\Delta^2}\rho^{6S-2T}\right)\cos^2\alpha = 0,
\end{equation}
which can be written as a quadratic equation in $\tan\alpha$
\begin{equation}
	\rho^2 \tan^2\alpha + 2T\tau\rho \tan\alpha + \left(T^2\tau^2 - \dsfrac{S^2}{\Delta^2}\rho^{6S-2T}\right) = 0,
\end{equation}
which leads to the solution
\begin{equation}
\label{eq_null_vectors}
	\tan\alpha_\pm = -\dsfrac{T\tau}{\rho} \pm \dsfrac{S}{\Delta}\rho^{3S-T-1}.
\end{equation}

Therefore, the incoming and outgoing null geodesics satisfy the differential equation
\begin{equation}
\label{eq_null_geodesics}
	\dsfrac{\de\tau}{\de\rho} = -\dsfrac{T\tau}{\rho} \pm \dsfrac{S}{\Delta}\rho^{3S-T-1}.
\end{equation}

The coordinate $\rho$ remains spacelike only as long as $g_{\rho\rho}>0$, and from equation \eqref{eq_rn_metric_tau_rho_matrix} we can see that this requires that
\begin{equation}
\label{eq_rho_spacelike_condition}
	\dsfrac{S^2}{\Delta^2}\rho^{6S-2T} > T^2\tau^2.
\end{equation}
To ensure the condition \eqref{eq_rho_spacelike_condition} in a neighborhood of $(0,0)$, we need to choose $T$ so that
\begin{equation}
\label{eq_rho_spacelike_condition_T}
	T \geq 3S.
\end{equation}

The null geodesics are the integral curves of the null vectors found in \eqref{eq_null_vectors}. We see that, in the coordinates $(\tau,\rho)$, the null geodesics are oblique everywhere, except at $\rho=0$, where they become tangent to the axis defined by $\rho=0$. Hence, the degeneracy of the metric is expressed by the fact that the lightcones stretch as approaching $\rho=0$, where they become degenerate (figure \ref{lightcones}). At these points, the incoming null geodesics become tangent to the outgoing null geodesics (figure \ref{null-geodesics}).

\image{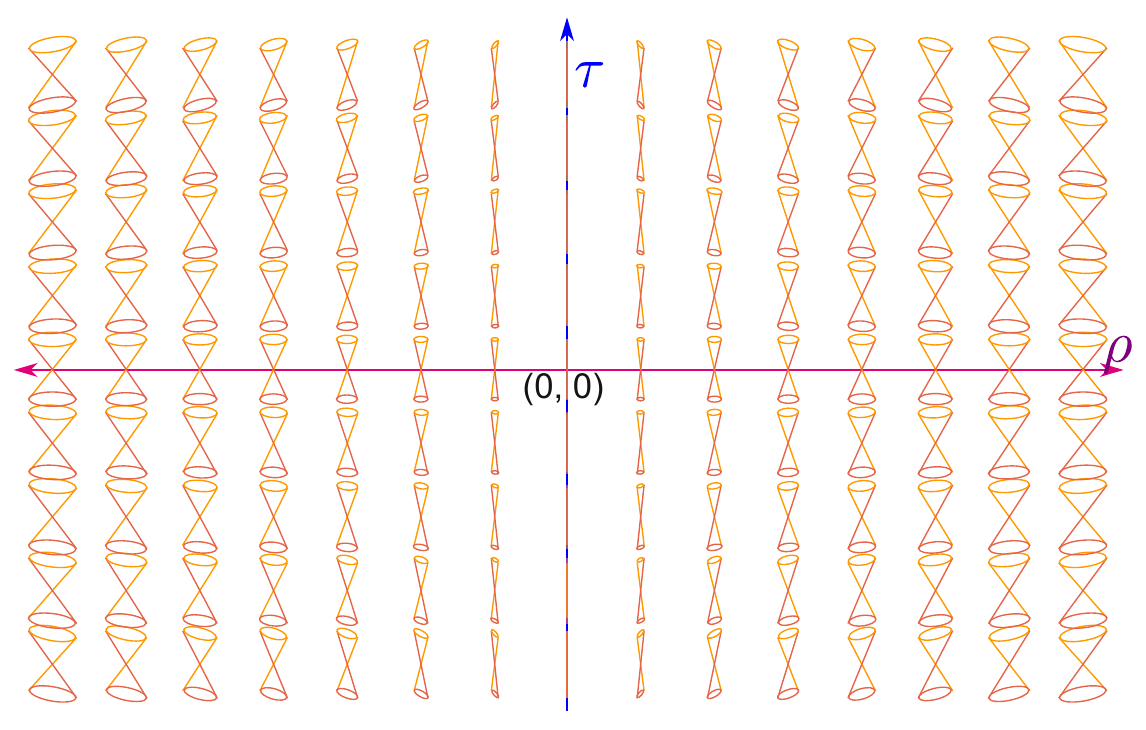}{0.7}{As one approaches the singularity on the axis $\rho=0$, the lightcones become more and more degenerate along that axis (for $T\geq 3S$ and even $S$).}

\image{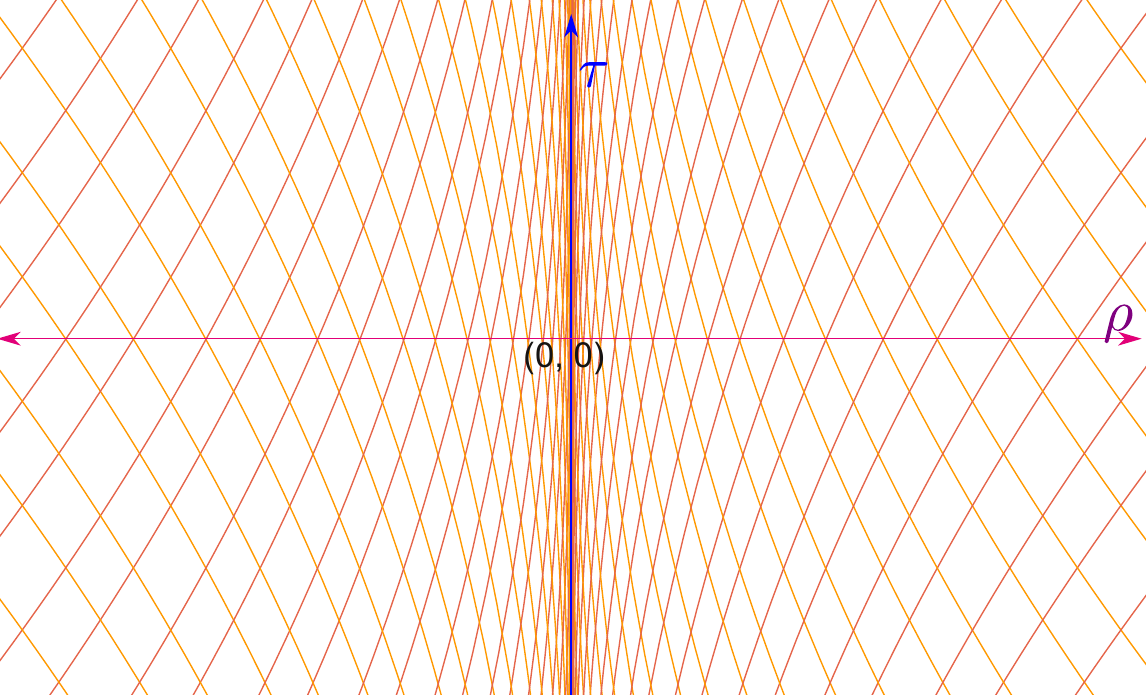}{0.7}{The null geodesics, in the $(\tau,\rho)$ coordinates, for $T\geq 3S$ and even $S$.}

%~~~~~~~~~~~~~~~~~~~~~~~~~~~~~~~~~~~~~~~~~~~~~~~~~~~~~~~~~~~~~~~~~~~~~~~%
\section{The Penrose-Carter diagrams for our solution}
\label{s_rn_ext_ext_penrose_carter}

To move to Penrose-Carter coordinates (and have a bird's eye view of the global behavior of the degenerate extensions of the {\rn} solution), we apply the same steps as those one normally applies for the standard {\rn} black hole. These steps are, for example, presented in \citep{HE95}{157-161}, and lead from the coordinates $(t,r)$ to the Penrose-Carter coordinates (figure \ref{std-rn}).

\image{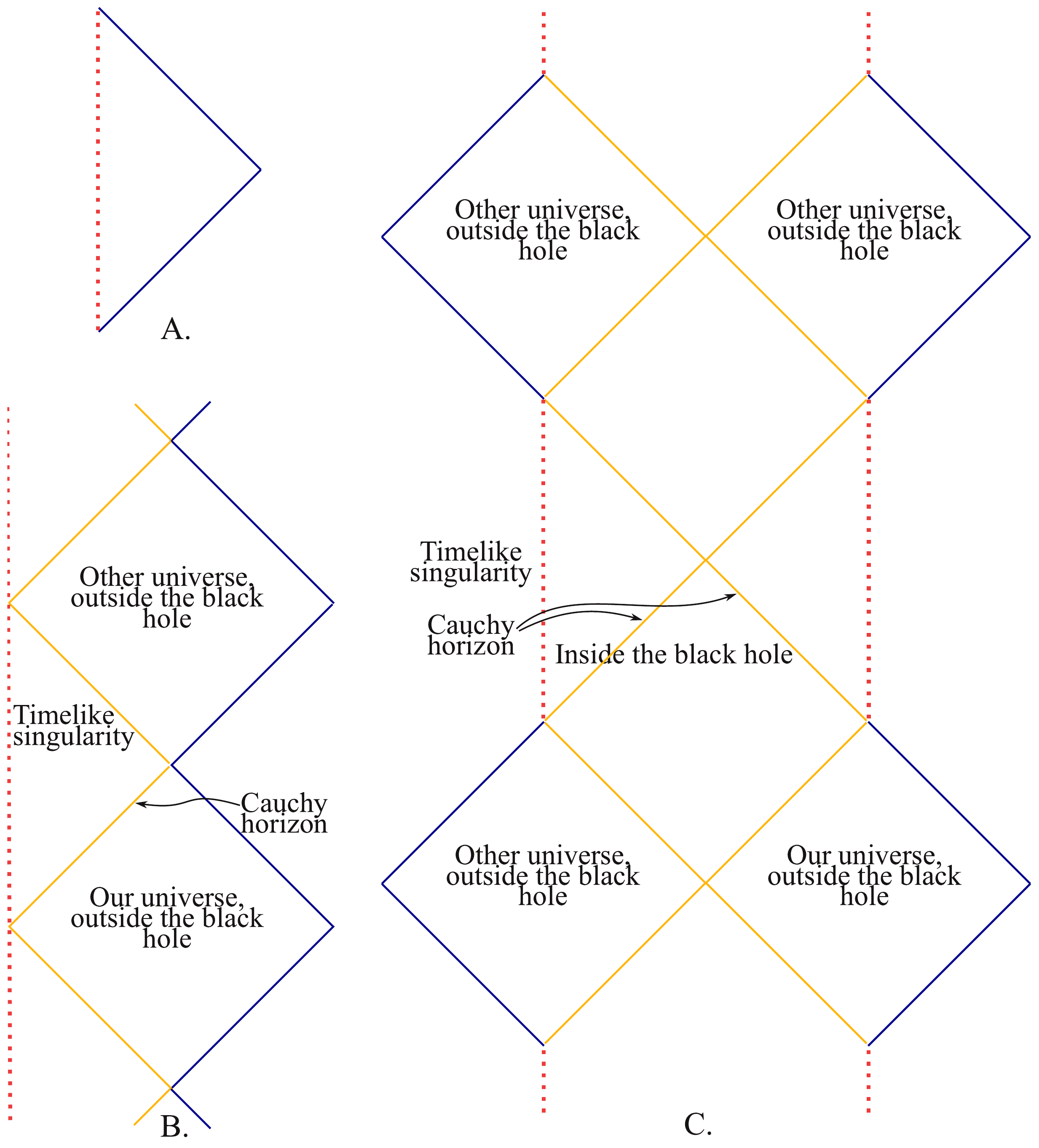}{0.6}{A. Naked Reissner-Nordstr\"om black holes ($q^2>m^2$). B. Extremal Reissner-Nordstr\"om black holes ($q^2=m^2$). C. Reissner-Nordstr\"om black holes with $q^2<m^2$.}

We just add our coordinate transformation before the steps leading to the Penrose-Carter coordinates, as we did in \cite{Sto11e} for the {\schw} solution. If $S$ is odd, the spacetime has a region $\rho<0$ and the Penrose-Carter diagrams are similar to the standard diagrams for the Kerr and Kerr-Newman spacetimes (see for example \citep{HE95}{165}). If $S$ is even, the diagram will repeat not only vertically, but also horizontally, symmetrical to the singularity.

We obtain the diagram for the naked {\rn} black hole ($q^2>m^2$) by taking the symmetric of the standard {\rn} diagram with respect to the singularity (figure \ref{ext-ext-rn-naked}).

\image{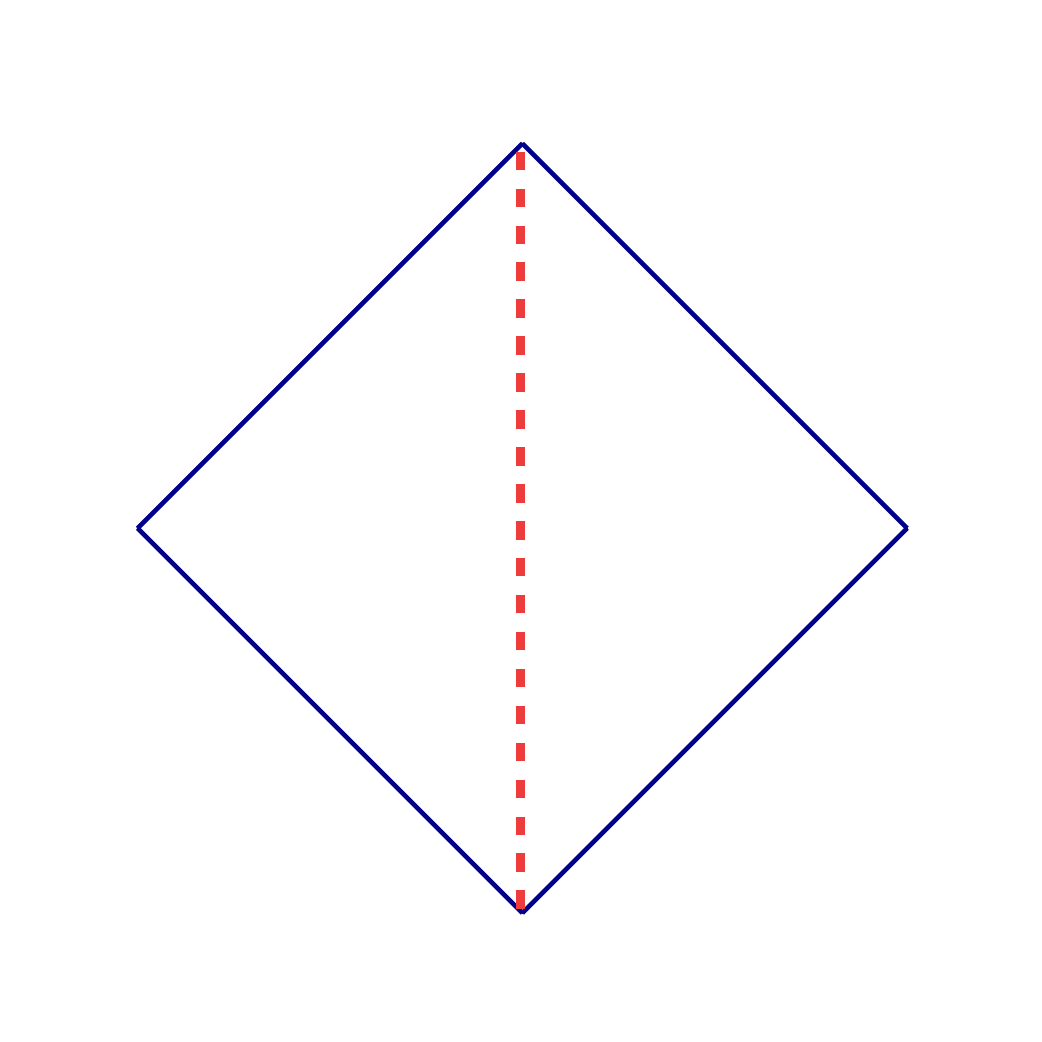}{0.3}{Penrose-Carter diagram for the naked {\rn} black hole ($q^2>m^2$), analytically extended beyond the singularity. It is symmetric with respect to the timelike singularity.}

The resulting diagram for the extremal {\rn} black hole ($q^2=m^2$) is a strip symmetric about the singularity (figure \ref{ext-ext-rn-extremal}). 

\image{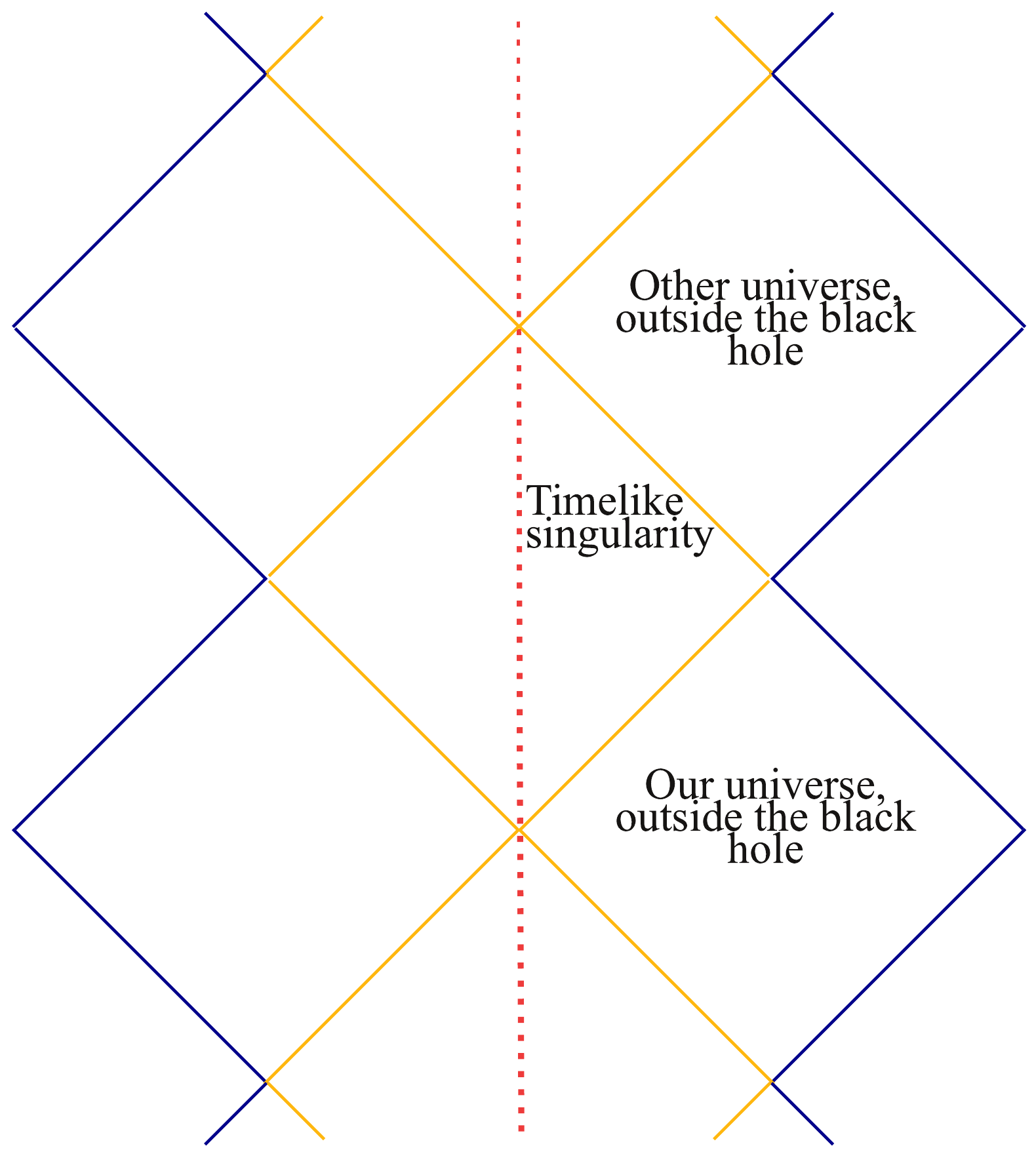}{0.5}{Penrose-Carter diagram for the extremal {\rn} black hole ($q^2=m^2$), analytically extended beyond the singularity. It repeats periodically along the vertical direction.}

When represented in plane, the diagram for the non-extremal {\rn} black hole ($q^2<m^2$) extends in two directions and has overlapping parts (figure \ref{ext-ext-rn}).

\image{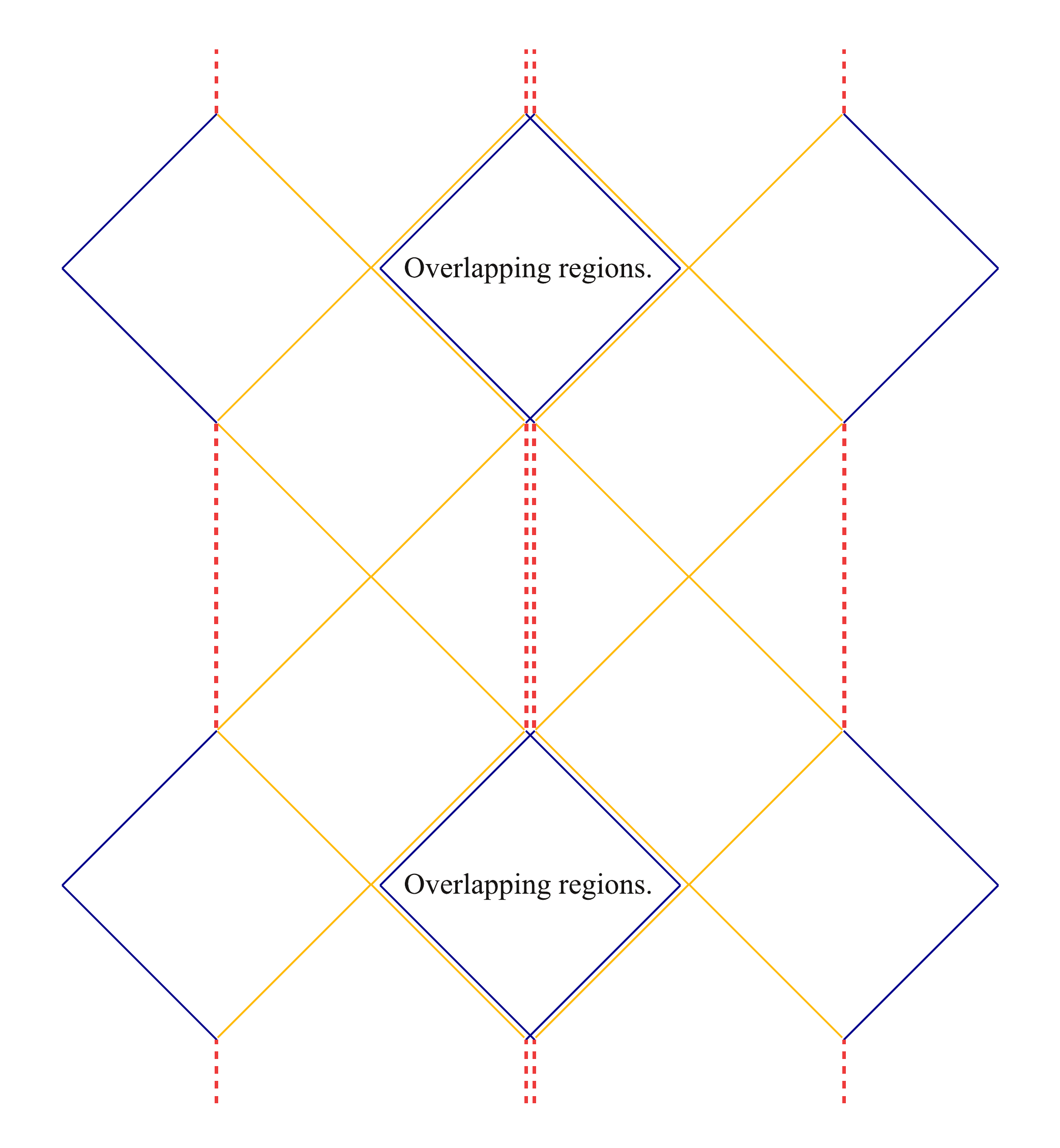}{0.6}{The Penrose-Carter diagram for the non-extremal {\rn} black hole with $q^2<m^2$, analytically extended beyond the singularity. When represented in plane, it repeats periodically along both the vertical and the horizontal directions, and it has overlaps. In the diagram, there is a small shift between the two copies, to make the overlapping visible.}

In the Penrose-Carter diagrams of the degenerate extension of the {\rn} solution the null geodesics continue through the singularity, because they are always at $\pm \frac {\pi}{4}$.

%~~~~~~~~~~~~~~~~~~~~~~~~~~~~~~~~~~~~~~~~~~~~~~~~~~~~~~~~~~~~~~~~~~~~~~~%
\section{A globally hyperbolic charged black hole}
\label{s_rn_globally_hyperbolic}

A global solution to the Einstein equation is well-behaved when the equations at a given moment of time determine the solution for the entire future and past. This condition is ensured by the \textit{global hyperbolicity}, which is expressed by the requirements that
\begin{enumerate}
	\item 
for any two points $p$ and $q$, the intersection between the causal future of $p$, and the causal past of $q$, $J^+(p)\cap J^-(q)$, is a compact subset of the spacetime;
	\item 
	there are no closed timelike curves (\citep{HE95}{206}).
\end{enumerate}

The property of global hyperbolicity is equivalent to the existence of a \textit{Cauchy hypersurface} -- a spacelike hypersurface $\mathfrak S$ that, for any point $p$ in the future (past) of $\mathfrak S$, is intersected by all past-directed (future-directed) inextensible causal (\ie timelike or null) curves through the point $p$ (\citep{HE95}{119, 209--212}).

Because in the standard coordinates for the {\rn} spacetime one cannot extend the solution beyond the singularity, the {\rn} spacetime fails to admit a Cauchy hypersurface, and it is normally inferred that it is not globally hyperbolic.

But since we now know how to extend analytically the {\rn} spacetime beyond the singularity, we should check if we can use this feature to construct new solutions which are globally hyperbolic. To do so, we will construct solutions that admit \textit{foliations} with Cauchy hypersurfaces -- \ie that are diffeomorphic with a Cartesian product between an interval $I\subseteq\R$ representing the time dimension, and a spacelike hypersurface.

The coordinates $(\tau,\rho)$, under the condition \eqref{eq_rho_spacelike_condition_T}, provide a spacelike foliation given by the hypersurfaces $\tau=\tn{const}$. This foliation is global only for naked singularities; otherwise it is defined locally, in a neighborhood of $(\tau,\rho)=(0,0)$ given by $r<r_-$. From the equation \eqref{eq_rn_metric} defining the {\rn} metric we know that the solution is \textit{stationary}; that is, when expressed in the coordinates $(t,r)$ it is independent of time. This means that we can choose as the origin of time any value, this ensuring that we can cover a neighborhood of the entire axis $\rho=0$ with coordinate patches like $(\tau,\rho)$. To obtain global foliations with Cauchy hypersurfaces, we use the global extensions represented in the Penrose-Carter diagrams of section \sref{s_rn_ext_ext_penrose_carter}, figures \ref{ext-ext-rn-naked}, \ref{ext-ext-rn-extremal} and \ref{ext-ext-rn}. It is important to note that in the Penrose-Carter diagram, the null directions are represented as straight lines inclined at $\pm \frac {\pi}{4}$.

For the naked {\rn} solution (figure \ref{ext-ext-rn-naked}) we can find immediately a global foliation, because the Penrose-Carter diagram is identical to that for the Minkowski spacetime (figure \ref{diamond-rn-naked}). Hence, the natural foliation of the Minkowski spacetime will be good for our extended naked {\rn} solution too.

To obtain explicitly the foliations for all the cases, we map to our solutions represented in coordinates $(\tau,\rho)$ the product $(0,1)\times\R$. To do this, we can use a version of the Schwarz-Christoffel mapping that maps the strip
\begin{equation}
\label{eq_strip}
\mc S:=\{z\in\C|\tn{Im}(z)\in[0,1]\}
\end{equation}
to a polygonal region from $\C$, with the help of the formula
\begin{equation}
	\label{eq_s_c_map}
	f(z)=A + C\int^{\mc S}\exp\left[\frac\pi 2(\alpha_--\alpha_+)\zeta\right]\prod_{k=1}^n\left[\sinh \frac\pi 2(\zeta-z_k)\right]^{\alpha_k-1}\de\zeta,
\end{equation}
where $z_k\in\partial\mc S:=\R\times\{0,i\}$ are the prevertices of the polygon, and $\alpha_-,\alpha_+,\alpha_k$ are the measures of the angles of the polygon, divided by $\pi$ (\cf \eg \cite{dri02}). The vertices having the angles $\alpha_-$ and $\alpha_+$ correspond to the ends of the strip, which are at infinity. The level curves $\{\tn{Im}(z)=\tn{const.}\}$ give our foliation \cite{Sto11c}.

The prevertices whose image is represented in Figure \ref{diamond-t} are
\begin{equation}
\left(-\infty,0, +\infty, i\right),
\end{equation}
and the angles are
\begin{equation}
\label{eq_angles_diamond-s}
\left(\frac {\pi}{2},\frac {\pi}{2},\frac {\pi}{2},\frac {\pi}{2}\right).
\end{equation}

\image{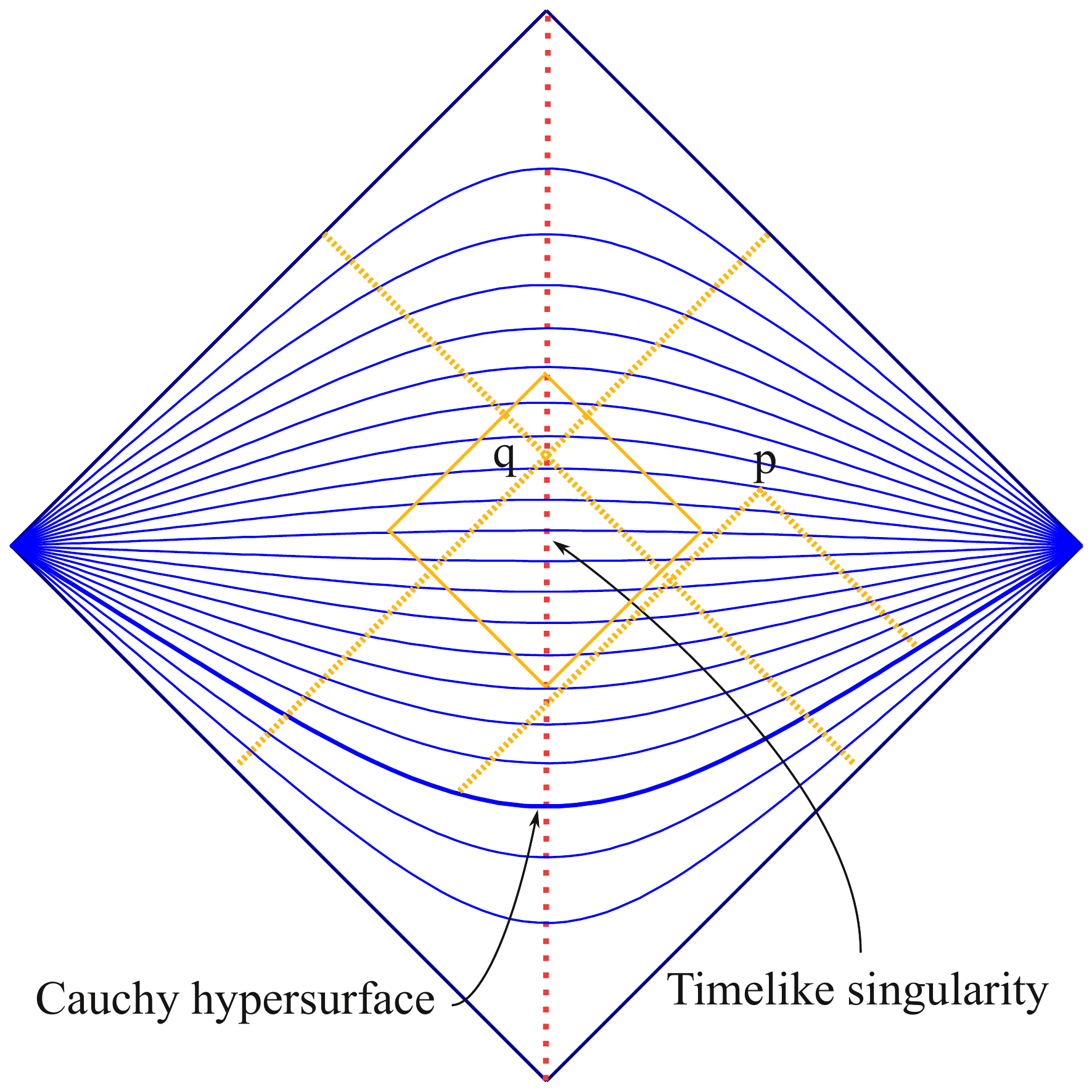}{0.5}{Spacelike foliation of the naked Reissner-Nordstr\"om solution ($q^2>m^2$). The spacelike hypersurfaces are Cauchy. Every point of the singularity can be joined with the future and past null infinities.}

For the other cases $q^2 \leq m^2$ (figure \ref{std-rn} (\textbf{B}) and (\textbf{C})) the maximal extensions cannot be globally hyperbolic, because they admit \textit{Cauchy horizons} (hypersurfaces which are boundaries for the Cauchy development of the data on a spacelike hypersurface). If we want to obtain a globally hyperbolic solution, we have to drop the regions beyond the Cauchy horizons. This leads naturally to a choice of a a subset of the Penrose-Carter diagram which is symmetric about the singularity $r=0$ and can be foliated (Figures \ref{up-big} and \ref{up-small-f}). 

\image{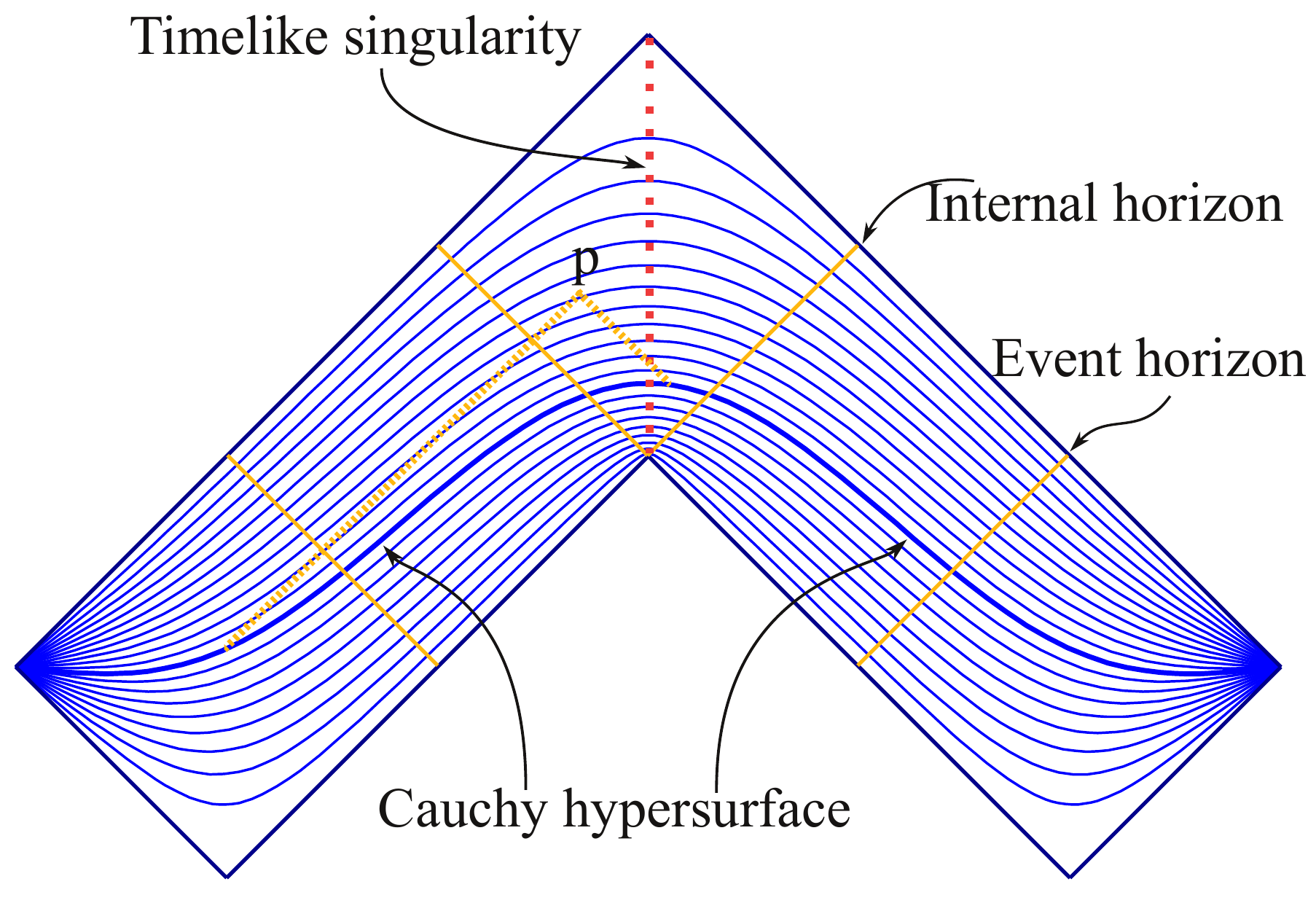}{0.6}{Foliation of the non-extremal Reissner-Nordstr\"om solution ($q^2<m^2$), with Cauchy hypersurfaces.}

\image{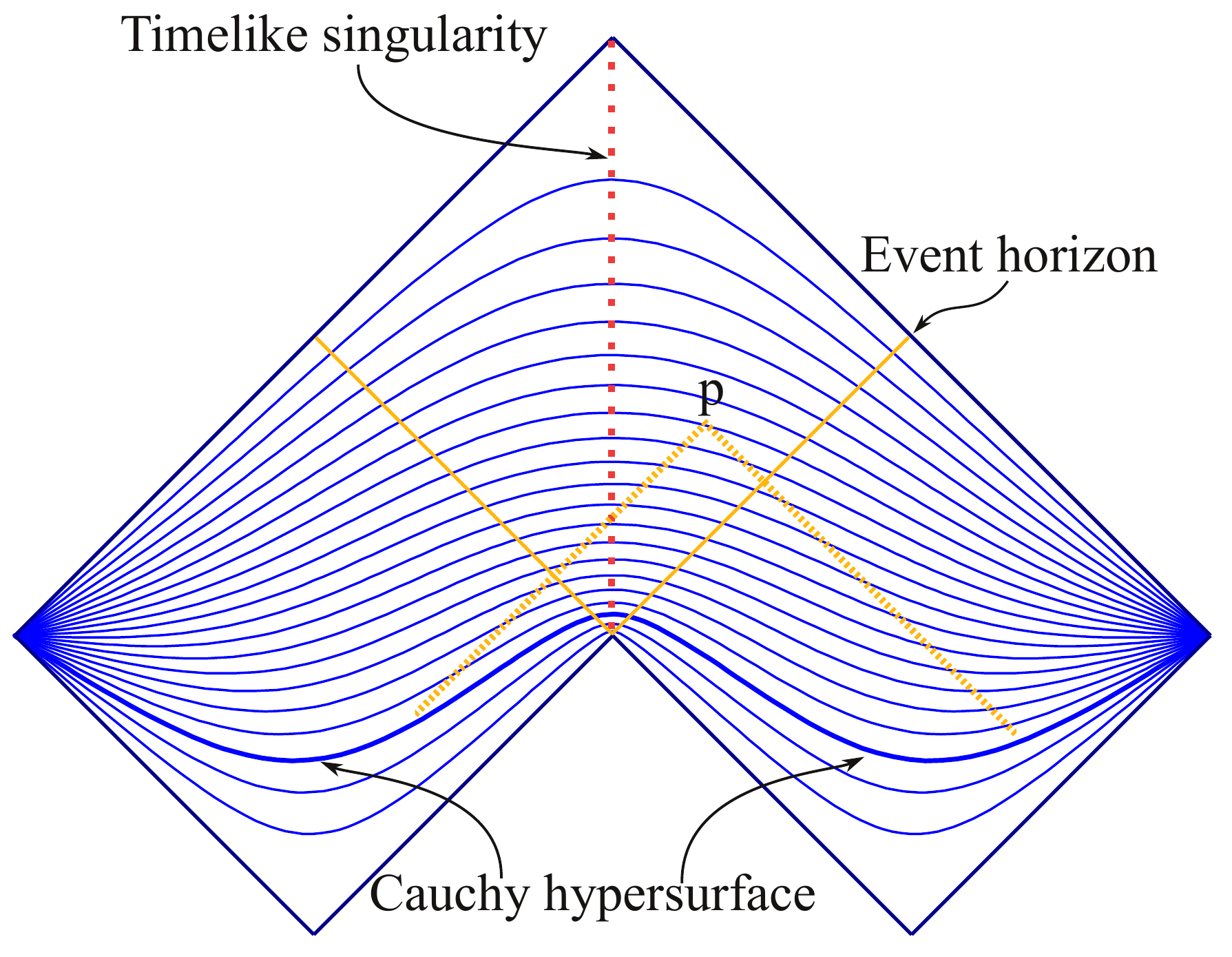}{0.4}{Foliation of the extremal Reissner-Nordstr\"om solution with $q^2=m^2$, with Cauchy hypersurfaces.}

Let us take now as prevertices of the Schwarz-Christoffel mapping \eqref{eq_s_c_map} the set
\begin{equation}
\label{eq_prevertices_rn-kerr}
\left(-\infty,-a, 0, a, +\infty, i\right),
\end{equation}
where $0<a$ is a positive real number. The angles are, respectively
\begin{equation}
\label{eq_angles_rn-kerr}
\left(\frac {\pi}{2},\frac {\pi}{2},\frac {3\pi}{2},\frac {\pi}{2},\frac {\pi}{2},\frac {\pi}{2}\right).
\end{equation}
Appropriate choices of $a$ result in the foliations represented in diagrams \ref{up-big} and \ref{up-small-f}, corresponding to the non-extremal, respectively the extremal solutions with $q^2<m^2$. Since $\alpha_-=\alpha_+$ and the edges are inclined at most at $\frac{\pi}{4}$, alternating in such a way that the level curves with $\tn{Im}(z)\in(0,1)$ have at each point tangents making an angle strictly between $-\frac{\pi}{4}$ and $\frac{\pi}{4}$, our foliations are spacelike.

In each of figures \ref{diamond-rn-naked}, \ref{up-big}, and \ref{up-small-f} we highlighted a spacelike hypersurface which is Cauchy, because it is intersected by all past (future) directed inextensible causal curves through a point $p$ from its future (past).

%~~~~~~~~~~~~~~~~~~~~~~~~~~~~~~~~~~~~~~~~~~~~~~~~~~~~~~~~~~~~~~~~~~~~~~~%
\section{The meaning of the analytic extension at the singularity}
\label{s_rn_analytic_meaning}

As in the case of the extension of the {\schw} solution, we can see that the singularity is not necessarily harmful for the information or the structure of spacetime. There is no reason to believe that the information is lost at the singularity, and the fact that it is timelike and may be naked, although it contradicts Penrose's cosmic censorship hypothesis, is compatible with the global hyperbolicity. These observations may apply also to the case of an evaporating charged black hole (see figure \ref{diamond-t}).

The {\rn} solution is, according to the no-hair theorem, representative of non-rotating and electrically charged black holes. If the black hole evaporates, the singularity becomes visible to the distant observers. This is a problem in the solutions which do not admit extension through the singularity. Our solution, because it can be extended beyond the singularity, does not break the topology of spacetime. The metric tensor does not run into infinities, although, because of its degeneracy, other quantities, such as its inverse, may become infinite.

The maximal globally hyperbolic extensions from section \sref{s_rn_globally_hyperbolic} are ideal, because the {\rn} solutions describe spacetimes which are too simple. But because they can be foliated by Cauchy hypersurfaces, and the base hypersurface is $\R^3$, we can interpolate between such solutions and foliations without singularities, and construct more general solutions. The interpolation can be done by varying the parameters $m$ and $q$. By this, one can model spacetimes with black holes that are formed and then evaporate. The presence of a timelike evaporating singularity of this type is compatible with the global hyperbolicity, as in figure \ref{diamond-t}.

\image{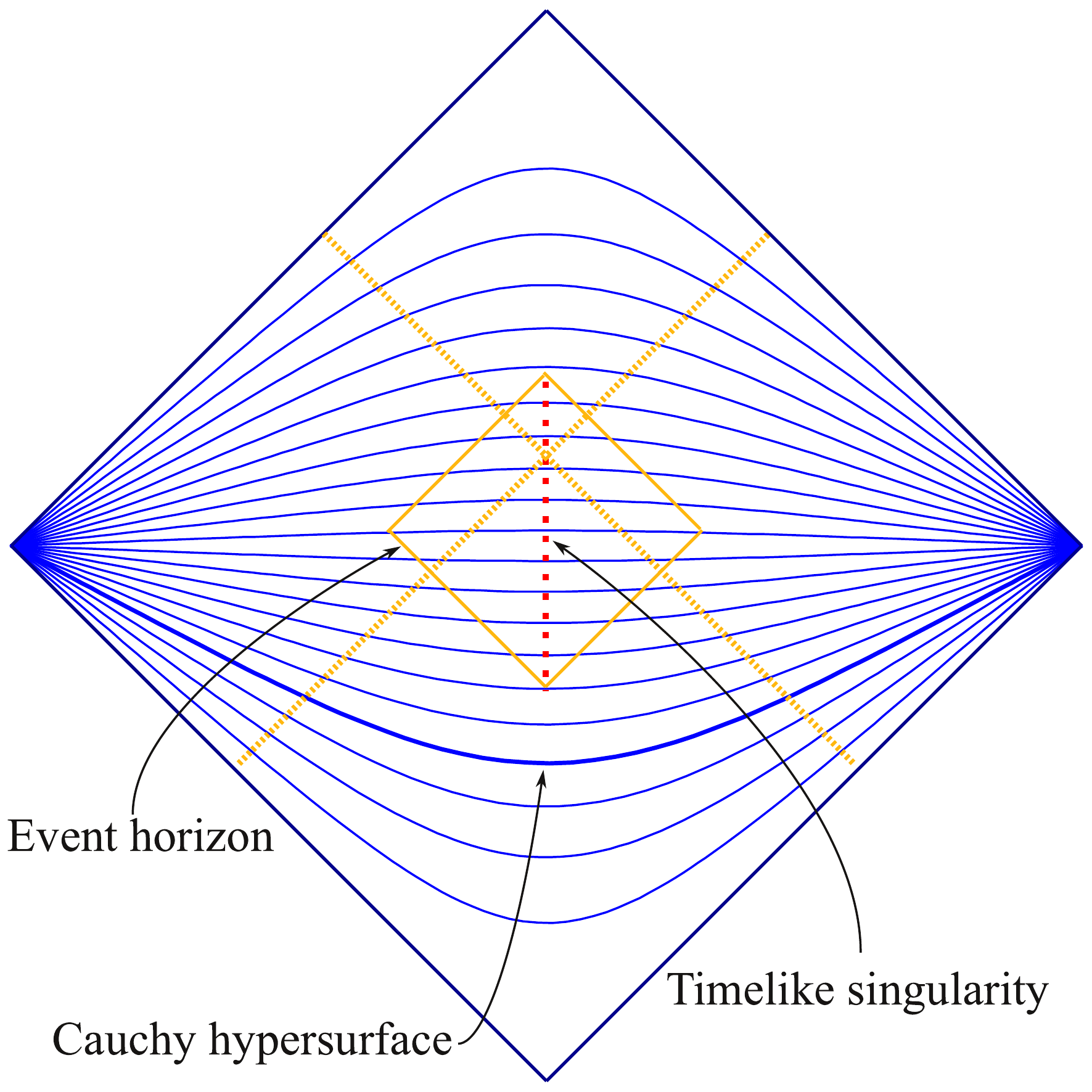}{0.5}{Non-primordial evaporating black hole with timelike singularity. The fact that the points of the singularity become visible to distant observers is not a problem for the global hyperbolicity, because the null geodesics can be extended beyond the singularity.}

This extension of the {\rn} solution can be used to model electrically charged particles as charged black holes, as pointed out in Remark \ref{rem_rn_wormhole}.

%~~~~~~~~~~~~~~~~~~~~~~~~~~~~~~~~~~~~~~~~~~~~~~~~~~~~~~~~~~~~~~~~~~~~~~~%
\textbf{Acknowledgments}

This work was partially supported by the Romanian Government grant PN II Idei 1187.

I thank an anonymous referee for the valuable comments and suggestions to improve the clarity and the quality of this paper.

\bibliographystyle{plain}%{unsrt}%{amsalpha}%{amsplain}

\end{document}